\documentclass[
  reprint,
  aps,
  prx,
  amsmath,
  amssymb,
  nofootinbib,
  superscriptaddress
]{revtex4-2}

\usepackage[T1]{fontenc}
\usepackage[utf8]{inputenc}
\usepackage{lmodern}
\usepackage{microtype}
\usepackage{mathtools}
\usepackage{bm}
\usepackage{booktabs}
\usepackage{graphicx}
\usepackage{xcolor}
\usepackage{hyperref}
\usepackage{enumitem}
\usepackage{amsthm}

\hypersetup{
  colorlinks=true,
  linkcolor=blue!45!black,
  citecolor=green!35!black,
  urlcolor=blue!55!black,
  pdftitle={Predictive Holonomy and Minimal Completion of Coarse Dynamics},
  pdfauthor={Maxim V. Churilov},
  pdfsubject={Coarse-grained dynamics, memory kernels, minimal realization, predictive state},
  pdfkeywords={coarse-graining, Mori-Zwanzig, Hankel rank, predictive state, minimal realization}
}

\newtheorem{theorem}{Theorem}
\newtheorem{proposition}[theorem]{Proposition}
\newtheorem{lemma}[theorem]{Lemma}
\newtheorem{corollary}[theorem]{Corollary}
\newtheorem{principle}[theorem]{Principle}
\theoremstyle{definition}
\newtheorem{definition}[theorem]{Definition}
\theoremstyle{remark}
\newtheorem{remark}[theorem]{Remark}

\newcommand{\R}{\mathbb{R}}

\newcommand{\one}{\mathbf{1}}
\newcommand{\rank}{\operatorname{rank}}
\newcommand{\im}{\operatorname{im}}
\newcommand{\kerop}{\operatorname{ker}}

\newcommand{\KL}{D_{\mathrm{KL}}}

\newcommand{\cM}{\mathcal{M}}
\newcommand{\cH}{\mathcal{H}}
\newcommand{\cS}{\mathcal{S}}

\newcommand{\cA}{\mathcal{A}}
\newcommand{\dd}{\,\mathrm{d}}
\newcommand{\e}{\mathrm{e}}
\newcommand{\norm}[1]{\left\lVert #1\right\rVert}

\begin{document}

\title{Predictive Holonomy and Minimal Completion of Coarse Dynamics}

\author{Maxim V. Churilov}
\email{churilovm1305@gmail.com}
\affiliation{Independent Researcher, Orenburg, Russia}

\date{July 13, 2026}

\begin{abstract}
A coarse description of a dynamical system generally fails to evolve autonomously: evolving and coarse-graining in different orders gives different predictions, and the compressed propagator fails the semigroup law.  We show that this path dependence is not merely an approximation error; it is the observable boundary of missing state.

For a finite continuous-time Markov chain with generator \(L\), coarse channel \(C\), and stochastic section \(R\) satisfying \(CR=I\), define \(P=RC\), \(Q=I-P\), and \(K=CLR\).  The exact coarse propagator \(T_t=C\e^{tL}R\) obeys a Volterra equation with memory kernel
\[
\cM(t)=CLQ\,\e^{tQLQ}\,QLR.
\]
First, we prove that every memory jet \(\cM^{(k)}(0)\) is recursively recoverable from the short-time derivatives of \(T_t\) alone.  Second, we form the block Hankel operator
\[
\mathbb H_\infty=\bigl(\cM^{(i+j)}(0)\bigr)_{i,j\ge0}
\]
and prove that its rank is exactly the smallest number of auxiliary linear state variables required to turn the coarse dynamics into an exact autonomous, time-homogeneous system.  A minimal completion is constructible by a Ho--Kalman factorization and is unique up to similarity.  Thus memory has a measurable state dimension.

The first nonzero memory jet also fixes the leading information-geometric failure of closure.  If \(d=\min\{k:\cM^{(k)}(0)\neq0\}\), then every smooth classical \(f\)-divergence generically scales as \(t^{2d+4}\).  We give exact Markov examples with exponents \(4,6,8\) and closure ranks \(1,2,3\), respectively, and verify the coefficients symbolically and at high precision.  An explicit entrywise-sampling bound turns the singular-value criterion into a finite-confidence certificate for hidden predictive coordinates.  Finally, we place the linear theorem inside the canonical nonlinear construction: histories are identified when they induce the same future law under every admissible intervention, yielding the minimal counterfactual predictive state.  The result converts coarse-graining from a choice of variables into a constructive state-completion problem and sharply separates ordinary linear realization from the harder requirement of a positive or Markov realization.
A matrix-Bernstein certificate recovers the causal rank directly from
structured samples with intrinsic-variance control and no entrywise
dimension loss.
\end{abstract}

\maketitle

\section{The missing-state problem}

The usual question of coarse-graining is which variables should be retained.  The more structural question is what must be added when the retained variables do not close.  Let a fine system be propagated and then compressed, or let its compressed description be propagated by an effective law.  If the two routes disagree, the discrepancy is often treated as model error.  We instead regard it as a boundary measurement of omitted state.

This viewpoint is forced by three classical bodies of work.  Markov lumpability characterizes exact autonomous aggregation \cite{KemenySnell1960,Buchholz1994,Kabernik2018}; projection methods turn eliminated variables into memory kernels \cite{Nakajima1958,Zwanzig1961,Mori1965,TeVrugt2020}; and realization theory identifies the smallest state-space model compatible with an input--output response \cite{Kalman1960,HoKalman1966,Kailath1980,Antoulas2005}.  Predictive-state and computational-mechanics constructions give an operational alternative to latent ontology: state is an equivalence class of histories with identical future consequences \cite{ShaliziCrutchfield2001,LittmanSuttonSingh2001,SinghJamesRudary2004,BarnettCrutchfield2015,ThonJaeger2015}.

The missing connection is exact and constructive.  A projected dynamics has a Mori--Zwanzig memory kernel, but the kernel is usually derived from a known microscopic generator.  A realization theorem reconstructs hidden state from a response sequence, but it is not normally identified with the state omitted by a specified coarse-graining.  Here these structures are joined:
\begin{equation}
\begin{aligned}
\text{coarse propagator}
&\longrightarrow \text{memory jets}
\longrightarrow \text{Hankel rank}\\
&\longrightarrow \text{minimal state completion}.
\end{aligned}
\label{eq:program}
\end{equation}
The first arrow is a recursive tomography theorem.  The second defines a closure invariant.  The third is a minimal-realization theorem.  Together they answer a question left open by a curvature-only diagnosis: not only whether a coarse law fails, but precisely how much hidden linear state is necessary and sufficient to repair it.

The paper has five principal results.

First, the exact compressed propagator
\[
T_t=C\e^{tL}R
\]
obeys a closed Volterra equation whose kernel factors into departure from the represented manifold, hidden propagation, and return to the observable derivative.

Second, all derivatives of that kernel are determined by derivatives of \(T_t\) at the origin.  No observation of the microscopic state is required once reset-conditioned coarse propagators can be measured.

Third, the infinite block Hankel rank of the memory jets equals the minimum dimension of any exact time-homogeneous linear auxiliary state.  A minimal completion is unique up to similarity, so the hidden coordinates are conventional while their dimension and response are invariant.

Fourth, the first nonzero memory jet determines the exponent and coefficient of short-time predictive divergence.  Memory depth and closure rank are independent invariants: the former measures when hidden state first returns; the latter measures how many independent hidden predictive modes exist.

Fifth, the linear theorem is the tangent-space sector of a more general operational statement.  The canonical exact state of a controlled stochastic process is the quotient of histories by equality of all intervention-conditioned future laws.  Every exact counterfactual representation factors through this quotient.

The requirement that a completion itself be probabilistic introduces a genuine frontier.  A minimal ordinary realization need not admit positive coordinates of the same dimension; positive realization can require more states and remains difficult in general \cite{FarinaRinaldi2000,BenvenutiFarina2004,TaghavianSjolund2025}.  We therefore distinguish three quantities: predictive-state dimension, linear closure rank, and positive Markov completion dimension.  Conflating them obscures the precise content of emergence.

\section{Operational state and predictive holonomy}

\subsection{Counterfactual predictive equivalence}

Consider a discrete-time controlled stochastic process with finite action alphabet \(\cA\) and observation alphabet \(\mathcal O\).  A history
\[
h_t=(o_0,a_0,o_1,\ldots,a_{t-1},o_t)
\]
is admissible if it has positive probability under at least one admissible policy.  A future test is a finite action string \(\bm a=a_t\cdots a_{t+\ell-1}\) together with an observation string \(\bm o=o_{t+1}\cdots o_{t+\ell}\).

\begin{definition}[Counterfactual predictive equivalence]
Two admissible histories \(h\) and \(h'\) are equivalent, written \(h\sim h'\), when
\begin{equation}
\Pr(\bm O=\bm o\mid h,\operatorname{do}(\bm A=\bm a))
=
\Pr(\bm O=\bm o\mid h',\operatorname{do}(\bm A=\bm a))
\label{eq:counterfactual-equivalence}
\end{equation}
for every finite future action string \(\bm a\) and every compatible future observation string \(\bm o\).
\end{definition}

The quotient \(\cS_\epsilon=\cH/\!\sim\) is the counterfactual predictive state space.  It extends causal states of autonomous processes \cite{ShaliziCrutchfield2001} and input--output causal states \cite{BarnettCrutchfield2015}; in predictive-state language, each quotient element is the complete vector of test predictions \cite{LittmanSuttonSingh2001,SinghJamesRudary2004}.

\begin{theorem}[Canonical predictive state]
\label{thm:canonical-state}
The quotient \(\cS_\epsilon\) has the following properties.

\begin{enumerate}[label=(\roman*)]
\item It is predictively sufficient for every finite intervention-conditioned future.
\item Its update is recursive: whenever \(hao\) is admissible, the map
\[
([h],a,o)\longmapsto[hao]
\]
is well defined.
\item It is minimal among exact counterfactual state representations.  If \(\phi(h)\) determines all probabilities in Eq.~\eqref{eq:counterfactual-equivalence}, then there exists a map \(g\) such that
\[
[h]=g(\phi(h)).
\]
\end{enumerate}
\end{theorem}

\begin{proof}
Sufficiency is the definition of the equivalence classes.  For recursive update, let \(h\sim h'\) and suppose both extensions by \(a,o\) are admissible.  Conditional probabilities of every later test after \(hao\) are ratios of joint test probabilities after \(h\); equality of all joint probabilities for \(h\) and \(h'\) gives equality after conditioning.  Hence \(hao\sim h'ao\).  For minimality, if \(\phi(h)=\phi(h')\) while \(h\not\sim h'\), then some intervention-conditioned future probability differs, contradicting exactness of \(\phi\).  Therefore each level set of \(\phi\) lies within one predictive-equivalence class, and the quotient label is a function of \(\phi\).
\end{proof}

Theorem~\ref{thm:canonical-state} is an operational closure theorem: exact state is not postulated as a hidden object but derived as the coarsest distinction preserving all possible future tests.  It is global and generally nonlinear.  The remaining sections analyze the finite-dimensional linear completion visible near a chosen coarse description.

\subsection{Holonomy as path dependence}

Let \(\Phi_\gamma\) and \(\Phi_{\gamma'}\) be two composite channels with the same operational endpoints.  For a divergence \(D\), define their predictive holonomy at state \(\rho\) by
\begin{equation}
\mathfrak h_D(\gamma,\gamma';\rho)
=
D\!\left(\Phi_\gamma(\rho)\,\middle\|\,\Phi_{\gamma'}(\rho)\right).
\label{eq:holonomy}
\end{equation}
The elementary time--scale cell compares
\[
\text{fine evolution}\ \to\ \text{coarse readout}
\]
with
\[
\text{coarse preparation}\ \to\ \text{effective evolution}.
\]
Vanishing holonomy means path-independent prediction.  Nonzero holonomy means that the coarse state fails to contain some information relevant to its claimed future law.

Equation~\eqref{eq:holonomy} is a scalar operational measure.  The more complete object is the channel defect itself.  For the Markov setting below, its finite-time composition law factorizes exactly through the unresolved subspace.  This factorization is the bridge from geometric language to state reconstruction.

\section{Exact coarse dynamics}

\subsection{Encoder, section, and projected generator}

Let
\[
\Delta_n=\{p\in\R^n:p_i\ge0,\ \one^\mathsf Tp=1\}
\]
and use the column convention
\[
\dot p(t)=Lp(t),
\]
where \(L_{ij}\ge0\) for \(i\neq j\) and \(\one^\mathsf TL=0\).  Let
\[
C:\Delta_n\to\Delta_m,\qquad
R:\Delta_m\to\Delta_n
\]
be stochastic matrices satisfying
\begin{equation}
CR=I_m.
\label{eq:section}
\end{equation}
The map \(C\) is the coarse channel; \(R\) specifies how a macrostate is physically prepared inside its microscopic fibre.  Define
\begin{equation}
P=RC,\qquad Q=I_n-P,\qquad K=CLR.
\label{eq:PQK}
\end{equation}
Then \(P^2=P\), \(Q^2=Q\), \(CQ=0\), and \(QR=0\).

\begin{proposition}
The matrix \(K=CLR\) is a valid continuous-time Markov generator on \(\Delta_m\).
\end{proposition}

\begin{proof}
Column conservation follows from stochasticity:
\[
\one_m^\mathsf TK=\one_n^\mathsf TLR=0.
\]
For \(\alpha\neq\beta\), the identity
\[
0=(CR)_{\alpha\beta}
=\sum_i C_{\alpha i}R_{i\beta}
\]
and nonnegativity imply \(C_{\alpha i}R_{i\beta}=0\) termwise.  Therefore diagonal entries of \(L\) make no contribution to \(K_{\alpha\beta}\), while all off-diagonal contributions are nonnegative.
\end{proof}

The reset-conditioned coarse propagator is
\begin{equation}
T_t=C\e^{tL}R.
\label{eq:Tt}
\end{equation}
It satisfies \(T_0=I_m\) and \(\dot T_0=K\).  Thus \(K\) is the unique linear generator matching the exact coarse response to first order for all macro preparations.

\subsection{Exact semigroup defect}

\begin{theorem}[Hidden-excursion factorization]
\label{thm:defect}
For all \(s,t\ge0\),
\begin{equation}
T_{s+t}-T_sT_t
=
C\e^{sL}Q\e^{tL}R.
\label{eq:defect}
\end{equation}
\end{theorem}

\begin{proof}
Insert \(I=P+Q\) between the two fine propagators:
\[
T_{s+t}=C\e^{sL}(P+Q)\e^{tL}R.
\]
The \(P\) term is
\[
C\e^{sL}RC\e^{tL}R=T_sT_t,
\]
leaving Eq.~\eqref{eq:defect}.
\end{proof}

The right-hand side has a literal dynamical meaning: evolve, retain only the component not represented by the coarse state, evolve that hidden component, and read out its later effect.  The failure of Chapman--Kolmogorov composition is therefore an observable hidden round trip.

\subsection{Memory equation}

Write the fine state as
\[
p(t)=Rq(t)+\eta(t),
\qquad
q(t)=Cp(t),\quad \eta(t)=Qp(t).
\]
For reset preparation \(p(0)=Rq_0\), one has \(\eta(0)=0\).  The coupled equations are
\begin{align}
\dot q(t)&=Kq(t)+CLQ\,\eta(t),\\
\dot\eta(t)&=QLR\,q(t)+QLQ\,\eta(t).
\end{align}
Define
\begin{equation}
A=CLQ,\qquad H=QLQ,\qquad B=QLR.
\label{eq:AHB}
\end{equation}

\begin{theorem}[Exact memory representation]
\label{thm:memory}
The coarse propagator satisfies
\begin{equation}
\dot T_t
=
KT_t+\int_0^t \cM(\tau)T_{t-\tau}\,\dd\tau,
\qquad T_0=I_m,
\label{eq:volterra}
\end{equation}
with
\begin{equation}
\cM(\tau)=A\e^{\tau H}B.
\label{eq:kernel}
\end{equation}
Its Laplace transform is
\begin{equation}
\widehat T(z)
=
\left[zI_m-K-A(zI-H)^{-1}B\right]^{-1}
\label{eq:resolvent}
\end{equation}
for sufficiently large \(\operatorname{Re}z\).
\end{theorem}

\begin{proof}
Variation of constants gives
\[
\eta(t)=\int_0^t\e^{(t-s)H}Bq(s)\,\dd s.
\]
Substitution into the equation for \(q\) yields Eq.~\eqref{eq:volterra}.  Laplace transformation gives
\[
z\widehat T-I=K\widehat T+\widehat{\cM}\widehat T,
\qquad
\widehat{\cM}(z)=A(zI-H)^{-1}B,
\]
and hence Eq.~\eqref{eq:resolvent}.
\end{proof}

The kernel separates three mechanisms:
\[
\text{resolved}
\xrightarrow{B}
\text{hidden}
\xrightarrow{\e^{tH}}
\text{hidden}
\xrightarrow{A}
\text{resolved derivative}.
\]
Memory is therefore not an abstract correction.  It is the transfer function of the state omitted by the coarse representation.

\section{Holonomy tomography}

The central experimental question is whether the hidden transfer function can be recovered without microscopic observation.  It can.

Let
\begin{equation}
S_n=T^{(n)}(0)=CL^nR,
\qquad
M_n=\cM^{(n)}(0)=AH^nB.
\label{eq:SnMn}
\end{equation}
The \(S_n\) are short-time response tensors of the measurable reset-conditioned coarse propagator.  The \(M_n\) are the unknown hidden memory moments.

\begin{theorem}[Memory-jet tomography]
\label{thm:tomography}
The full sequence \(\{M_n\}_{n\ge0}\) is uniquely determined by \(\{S_n\}_{n\ge0}\).  Specifically,
\begin{equation}
S_{n+1}
=
KS_n+\sum_{j=0}^{n-1}M_jS_{n-1-j},
\qquad n\ge1,
\label{eq:recurrence}
\end{equation}
and therefore
\begin{equation}
M_{n-1}
=
S_{n+1}-KS_n
-\sum_{j=0}^{n-2}M_jS_{n-1-j}.
\label{eq:tomography}
\end{equation}
\end{theorem}

\begin{proof}
Expand
\[
T_t=\sum_{n\ge0}\frac{S_n}{n!}t^n,
\qquad
\cM(t)=\sum_{j\ge0}\frac{M_j}{j!}t^j.
\]
The convolution term in Eq.~\eqref{eq:volterra} is
\begin{align}
\int_0^t\cM(\tau)T_{t-\tau}\,\dd\tau
&=
\sum_{j,k\ge0}
\frac{M_jS_k}{j!\,k!}
\int_0^t\tau^j(t-\tau)^k\,\dd\tau\\
&=
\sum_{j,k\ge0}
M_jS_k\,\frac{t^{j+k+1}}{(j+k+1)!},
\end{align}
because the beta integral cancels the factorials.  Matching the coefficient of \(t^n/n!\) in Eq.~\eqref{eq:volterra} gives Eq.~\eqref{eq:recurrence}; solving for the newest moment gives Eq.~\eqref{eq:tomography}.
\end{proof}

The first moments are
\begin{align}
M_0&=S_2-K^2,\\
M_1&=S_3-KS_2-M_0K,\\
M_2&=S_4-KS_3-M_0S_2-M_1K.
\end{align}
Thus the entire hidden response is encoded in the temporal jet of the observable coarse channel.

\begin{remark}
The tomography requires a preparation protocol \(R\): after each trial, the same macrostate must be reset to the same fine distribution.  Without controlled preparation, passive trajectories mix dynamical memory with uncertainty about the within-fibre state.
\end{remark}

\section{Minimal completion theorem}

\subsection{Closure rank}

Form the block Hankel matrices
\begin{equation}
\mathbb H_{p,q}
=
\begin{pmatrix}
M_0&M_1&\cdots&M_{q-1}\\
M_1&M_2&\cdots&M_q\\
\vdots&\vdots&\ddots&\vdots\\
M_{p-1}&M_p&\cdots&M_{p+q-2}
\end{pmatrix}
\label{eq:block-hankel}
\end{equation}
and the infinite Hankel operator
\[
\mathbb H_\infty=(M_{i+j})_{i,j\ge0}.
\]

\begin{definition}[Closure rank]
The linear closure rank is
\begin{equation}
r_{\mathrm{cl}}=\rank\mathbb H_\infty.
\label{eq:closure-rank}
\end{equation}
\end{definition}

Because \(M_k=AH^kB\), define the finite observability and controllability blocks
\begin{align}
\mathcal O_p&=
\begin{pmatrix}
A\\AH\\ \vdots\\AH^{p-1}
\end{pmatrix},&
\mathcal C_q&=
\begin{pmatrix}
B&HB&\cdots&H^{q-1}B
\end{pmatrix}.
\end{align}
Then
\begin{equation}
\mathbb H_{p,q}=\mathcal O_p\mathcal C_q.
\label{eq:hankel-factor}
\end{equation}

\subsection{Exact auxiliary completion}

A linear auxiliary completion of dimension \(r\) is a system
\begin{equation}
\frac{\dd}{\dd t}
\begin{pmatrix}q\\z\end{pmatrix}
=
\begin{pmatrix}
K&A_\star\\
B_\star&H_\star
\end{pmatrix}
\begin{pmatrix}q\\z\end{pmatrix},
\qquad
z(0)=0,
\label{eq:completion}
\end{equation}
whose \(q\)-component equals \(T_tq_0\) for every initial macrostate \(q_0\).

\begin{theorem}[Minimal completion]
\label{thm:minimal-completion}
Assume finite-dimensional fine dynamics as above.  Then:

\begin{enumerate}[label=(\roman*)]
\item An exact linear auxiliary completion exists.
\item The minimum auxiliary dimension is \(r_{\mathrm{cl}}\).
\item Every minimal completion realizes the same memory moments,
\[
M_k=A_\star H_\star^kB_\star,
\]
and any two controllable and observable minimal completions are related by an invertible hidden-state similarity.
\end{enumerate}
\end{theorem}

\begin{proof}
Existence of some finite completion follows from the microscopic triple \((A,H,B)\), whose effective dimension is at most \(\rank Q\).

For the lower bound, suppose Eq.~\eqref{eq:completion} of dimension \(r\) reproduces \(T_t\).  Eliminating \(z\) gives
\[
\cM_\star(t)=A_\star\e^{tH_\star}B_\star.
\]
Equality of the resulting coarse resolvent with Eq.~\eqref{eq:resolvent} for all sufficiently large \(z\) implies equality of the memory transfer functions and hence
\[
M_k=A_\star H_\star^kB_\star
\quad (k\ge0).
\]
Every finite Hankel block then factors through an \(r\)-dimensional space:
\[
\mathbb H_{p,q}
=
\begin{pmatrix}
A_\star\\A_\star H_\star\\ \vdots
\end{pmatrix}
\begin{pmatrix}
B_\star&H_\star B_\star&\cdots
\end{pmatrix}.
\]
Thus \(r_{\mathrm{cl}}\le r\).

For the matching upper bound, construct the state space directly from the infinite Hankel operator.  For \(j\ge0\) and \(u\in\R^m\), let
\[
c_{j,u}:=(M_{i+j}u)_{i\ge0}
\]
be an infinite block column, and let
\[
\mathcal X=\operatorname{span}\{c_{j,u}:j\ge0,\ u\in\R^m\}.
\]
By definition, \(\dim\mathcal X=r_{\mathrm{cl}}\).  Define
\[
B_\star u=c_{0,u},\qquad
H_\star c_{j,u}=c_{j+1,u},
\]
and let \(A_\star:\mathcal X\to\R^m\) extract the zeroth block component.  The shift \(H_\star\) is well defined: if
\[
\sum_\alpha c_{j_\alpha,u_\alpha}=0,
\]
then for every \(i\ge0\),
\[
\sum_\alpha M_{i+j_\alpha}u_\alpha=0;
\]
replacing \(i\) by \(i+1\) proves
\[
\sum_\alpha c_{j_\alpha+1,u_\alpha}=0.
\]
By construction,
\[
A_\star H_\star^kB_\star u=M_ku
\]
for every \(k\) and \(u\).  Hence \((A_\star,H_\star,B_\star)\) realizes the full analytic memory kernel in dimension \(r_{\mathrm{cl}}\).  Substitution into Eq.~\eqref{eq:resolvent} shows that the augmented system reproduces \(T_t\) exactly.  The lower bound is therefore attained.

The constructed realization is reachable because \(\mathcal X\) is spanned by \(H_\star^kB_\star u\), and observable because a vector whose every block component \(A_\star H_\star^\ell x\) vanishes is the zero infinite column.  For two reachable and observable minimal realizations, mapping states generated by identical finite input sequences defines an invertible intertwiner \(S\), yielding
\[
H_2=SH_1S^{-1},\qquad
B_2=SB_1,\qquad
A_2=A_1S^{-1}.
\]
\end{proof}

Theorem~\ref{thm:minimal-completion} is the main closure result.  The omitted state is not unique as a coordinate system, but its minimum dimension and input--output action are unique.  Memory has become state.

\begin{corollary}[Finite tomography]
\label{cor:finite-tomography}
If \(r_{\mathrm{cl}}=r\), then some finite pair \(p,q\le r\) has stabilized Hankel rank \(r\), and the corresponding Hankel and shifted Hankel blocks determine a minimal completion.  In particular, the square choice \(p=q=r\) requires at most
\[
M_0,\ldots,M_{2r-1},
\]
which are recoverable from
\[
S_0,\ldots,S_{2r+1}.
\]
\end{corollary}

For a stabilized finite block, choose a rank factorization
\[
\mathbb H_{p,q}=\mathcal O\mathcal C
\]
and let \(\mathbb H^+_{p,q}=(M_{i+j+1})\) be its one-step shift.  In any full-rank coordinates, the hidden generator is obtained by the Ho--Kalman formula
\begin{equation}
H_\star=\mathcal O^\dagger\mathbb H^+_{p,q}\mathcal C^\dagger,
\label{eq:hokalman}
\end{equation}
while \(A_\star\) and \(B_\star\) are the first block row and column of \(\mathcal O\) and \(\mathcal C\).

\begin{corollary}[Closure dichotomy]
\label{cor:closure}
The following are equivalent:
\begin{enumerate}[label=(\roman*)]
\item \(r_{\mathrm{cl}}=0\);
\item \(\cM(t)\equiv0\);
\item \(T_t=\e^{tK}\) for all \(t\ge0\);
\item \(T_{s+t}=T_sT_t\) for all \(s,t\ge0\).
\end{enumerate}
\end{corollary}

\begin{proof}
Zero Hankel rank is equivalent to \(M_k=0\) for every \(k\), hence to the vanishing analytic kernel.  The Volterra equation then reduces to \(\dot T=KT\).  Conversely, an exact semigroup has no memory self-energy in Eq.~\eqref{eq:resolvent}.  The last equivalence follows from differentiability and \(\dot T_0=K\).
\end{proof}

\subsection{Depth and rank are distinct}

Define the memory depth
\begin{equation}
d_\star=\min\{k\ge0:M_k\neq0\},
\label{eq:depth}
\end{equation}
with \(d_\star=\infty\) under exact closure.  Depth is a temporal invariant: it counts how many hidden generator insertions are required before unresolved information returns.  Closure rank is a dimensional invariant: it counts independent hidden predictive modes.  Neither determines the other.

The distinction is analogous to delay versus state dimension in realization theory.  A single mode can return immediately; many modes can cancel through several derivatives before becoming visible.

\section{Information geometry of the first return}

Let \(f:(0,\infty)\to\R\) satisfy
\[
f(1)=f'(1)=0,\qquad f''(1)>0,
\]
and define
\[
D_f(p\|q)=\sum_iq_i f(p_i/q_i).
\]
For an interior probability vector \(q\), the Fisher tangent norm is
\[
\norm{v}_{F(q)}^2=\sum_i\frac{v_i^2}{q_i}.
\]

\begin{lemma}[Local divergence Hessian]
\label{lem:fisher}
If \(p(t)\) and \(r(t)\) are probability curves with \(p(0)=r(0)=q\) and
\[
p(t)-r(t)=\frac{t^\ell}{\ell!}v+O(t^{\ell+1}),
\]
then
\begin{equation}
D_f(p(t)\|r(t))
=
\frac{f''(1)}{2(\ell!)^2}
t^{2\ell}\norm{v}_{F(q)}^2
+O(t^{2\ell+1}).
\label{eq:local-f}
\end{equation}
\end{lemma}

\begin{proof}
Write \(\delta_i=p_i-r_i\).  Taylor expansion at equality gives
\[
D_f(p\|r)
=
\frac{f''(1)}2\sum_i\frac{\delta_i^2}{r_i}
+O(\norm{\delta}^3).
\]
Since \(r_i=q_i+O(t)\), substitution of the assumed displacement yields Eq.~\eqref{eq:local-f}.
\end{proof}

\begin{theorem}[Curvature-depth law]
\label{thm:depth-law}
If \(d_\star<\infty\), then
\begin{equation}
T_t-\e^{tK}
=
\frac{t^{d_\star+2}}{(d_\star+2)!}M_{d_\star}
+O(t^{d_\star+3}).
\label{eq:path-depth}
\end{equation}
For every interior \(q\notin\kerop M_{d_\star}\),
\begin{equation}
D_f(T_tq\|\e^{tK}q)
=
\frac{f''(1)t^{2d_\star+4}}
{2[(d_\star+2)!]^2}
\norm{M_{d_\star}q}_{F(q)}^2
+O(t^{2d_\star+5}).
\label{eq:depth-law}
\end{equation}
\end{theorem}

\begin{proof}
In Eq.~\eqref{eq:volterra}, all kernel derivatives below order \(d_\star\) vanish, so
\[
\cM(\tau)=\frac{\tau^{d_\star}}{d_\star!}M_{d_\star}
+O(\tau^{d_\star+1}).
\]
At leading order \(T_{t-\tau}=I+O(t)\).  Integrating once in the Volterra equation gives Eq.~\eqref{eq:path-depth}; Lemma~\ref{lem:fisher} gives Eq.~\eqref{eq:depth-law}.
\end{proof}

The theorem makes the earliest hidden return experimentally visible as an exponent.  The full Hankel spectrum then determines how many hidden modes are required to complete the dynamics.

\section{An exact curvature--completion ladder}

We now give three finite Markov chains whose memory depths are \(0,1,2\) and whose closure ranks are \(1,2,3\).  Every matrix entry and coefficient below is exact.

\subsection{\texorpdfstring{Immediate return: depth \(0\), rank \(1\)}{Immediate return: depth 0, rank 1}}

Let
\begin{align}
L_0&=
\begin{pmatrix}
-2&1&1\\
1&-3&1\\
1&2&-2
\end{pmatrix},\\
C_0&=
\begin{pmatrix}
1&1&0\\
0&0&1
\end{pmatrix},
\qquad
R_0=
\begin{pmatrix}
1/2&0\\
1/2&0\\
0&1
\end{pmatrix}.
\end{align}
Then
\[
K_0=
\begin{pmatrix}
-3/2&2\\
3/2&-2
\end{pmatrix},
\qquad
M_0^{(0)}=
\begin{pmatrix}
1/4&0\\
-1/4&0
\end{pmatrix},
\]
and the memory Hankel has rank \(1\).  For \(q=(3/5,2/5)^\mathsf T\),
\begin{equation}
\KL(T_tq\|\e^{tK_0}q)
=
\frac{3}{256}t^4+O(t^5).
\end{equation}

\subsection{\texorpdfstring{One hidden step: depth \(1\), rank \(2\)}{One hidden step: depth 1, rank 2}}

Let
\begin{equation}
L_1=
\begin{pmatrix}
-23/8&7/8&5/4&3/4\\
9/8&-25/8&3/4&5/4\\
11/8&13/8&-7/2&1/2\\
3/8&5/8&3/2&-5/2
\end{pmatrix}
\end{equation}
and partition the states into \(\{1,2\}\) and \(\{3,4\}\), with uniform reconstruction in each block.  Then
\[
K_1=
\begin{pmatrix}-2&2\\2&-2\end{pmatrix},
\qquad M_0^{(1)}=0,
\]
while
\[
M_1^{(1)}
=
\begin{pmatrix}
1/8&-1/8\\
-1/8&1/8
\end{pmatrix}.
\]
The closure rank is \(2\).  For \(q=(7/10,3/10)^\mathsf T\),
\begin{equation}
\KL(T_tq\|\e^{tK_1}q)
=
\frac{1}{6048}t^6+O(t^7).
\end{equation}

\subsection{\texorpdfstring{Two hidden steps: depth \(2\), rank \(3\)}{Two hidden steps: depth 2, rank 3}}

The depth-two example is
\begin{widetext}
\begin{equation}
L_2=
\begin{pmatrix}
-52/15&23/15&23/15&1/6&1/10\\
19/15&-59/15&16/15&17/30&1/2\\
6/5&7/5&-18/5&11/30&3/10\\
17/30&17/30&11/30&-61/20&41/20\\
13/30&13/30&19/30&39/20&-59/20
\end{pmatrix}.
\label{eq:L2}
\end{equation}
\end{widetext}
Take blocks \(\{1,2,3\}\) and \(\{4,5\}\), with uniform reconstruction.  Then
\[
K_2=
\begin{pmatrix}-1&1\\1&-1\end{pmatrix},
\qquad
M_0^{(2)}=M_1^{(2)}=0,
\]
and
\begin{equation}
M_2^{(2)}
=
\frac1{625}
\begin{pmatrix}1&-1\\-1&1\end{pmatrix}.
\end{equation}
The exact memory-Hankel rank is \(3\).  For \(q=(3/5,2/5)^\mathsf T\),
\begin{equation}
\KL(T_tq\|\e^{tK_2}q)
=
\frac{1}{2\,700\,000\,000}t^8+O(t^9).
\label{eq:t8}
\end{equation}

The three examples establish that the first-return exponent and the minimal completion dimension form a genuine two-coordinate classification.  They are not different normalizations of one scalar defect.

\begin{figure}
\includegraphics[width=\columnwidth]{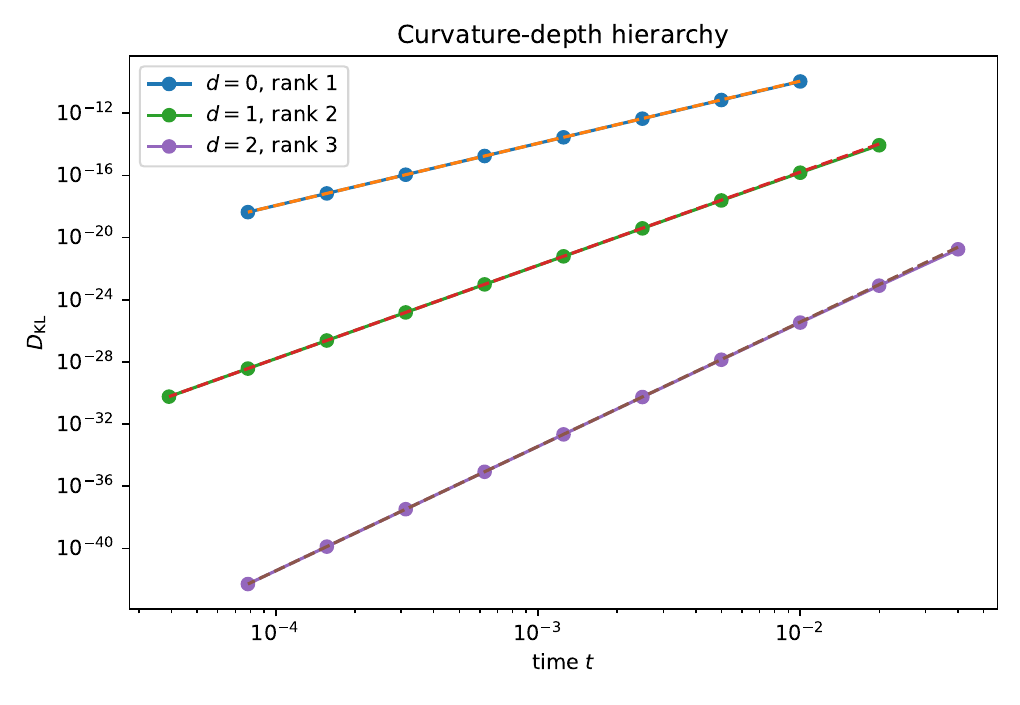}
\caption{High-precision Kullback--Leibler holonomy for the exact depth-\(0,1,2\) Markov examples.  The slopes approach \(4,6,8\), respectively, as predicted by Theorem~\ref{thm:depth-law}.}
\label{fig:depth}
\end{figure}

\begin{figure}
\includegraphics[width=\columnwidth]{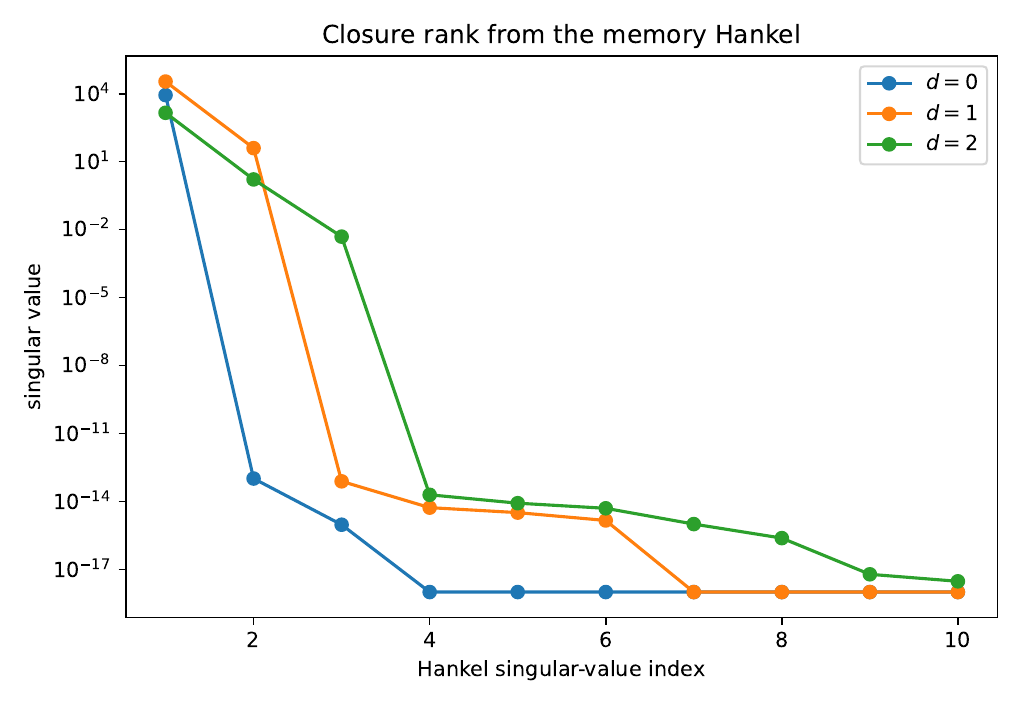}
\caption{Singular spectra of finite memory-Hankel blocks.  The exact ranks are \(1,2,3\); the remaining singular values are numerical roundoff.}
\label{fig:rank}
\end{figure}

\begin{figure}
\includegraphics[width=\columnwidth]{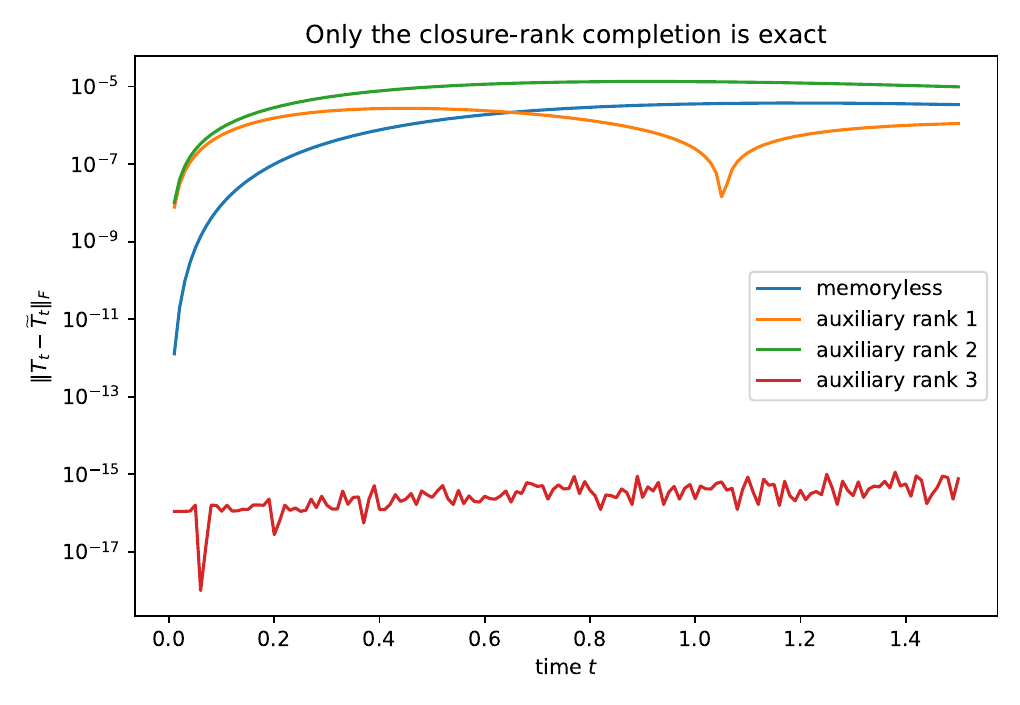}
\caption{Propagator error for truncated completions of the rank-three example.  Only the three-dimensional auxiliary realization reproduces the coarse propagator to machine precision for all tested times.}
\label{fig:completion}
\end{figure}

\section{From measured holonomy to hidden state}

Theorems~\ref{thm:tomography} and \ref{thm:minimal-completion} give a finite reconstruction protocol.

\begin{enumerate}[label=\arabic*.]
\item Prepare each macro basis state through the fixed section \(R\).
\item Estimate \(T_t\) for short times and infer derivatives \(S_0,\ldots,S_{2r+1}\).
\item Recover \(M_0,\ldots,M_{2r-1}\) recursively from Eq.~\eqref{eq:tomography}.
\item Form \(\mathbb H_{r,r}\) and increase \(r\) until its rank stabilizes under one extra block row and column.
\item Factor \(\mathbb H_{r,r}=\mathcal O\mathcal C\), form the shifted block, and compute \(H_\star\) from Eq.~\eqref{eq:hokalman}.
\item Read \(A_\star\) and \(B_\star\) from the first block row and column.
\end{enumerate}

The output is not a reconstruction of the microscopic coordinates.  It is the unique minimal hidden response class needed by the chosen coarse variables and preparation protocol.  Any microscopic degrees of freedom that are unreachable from \(B\) or unobservable through \(A\) disappear automatically.

This distinction matters.  The hidden subspace \(\im Q\) may be enormous, while the closure rank is small:
\begin{align}
r_{\mathrm{cl}}
&\le \dim\operatorname{span}\{H^k\im B:k\ge0\},\\
r_{\mathrm{cl}}
&\le \dim\operatorname{span}\{(H^\mathsf T)^k\im A^\mathsf T:k\ge0\},\\
r_{\mathrm{cl}}&\le\rank Q.
\end{align}
The completion captures only the controllable and observable part of the omitted dynamics.

\subsection{Noise and approximate rank}

With finite data, Hankel matrices have no exact numerical rank.  Their singular values provide a hierarchy of approximate closures.  Balanced truncation and Hankel-norm approximation offer natural error controls for stable linear systems \cite{Antoulas2005}.  In the present setting, the singular spectrum has a direct interpretation: it orders hidden modes by the strength with which coarse preparation can excite them and coarse readout can detect their return.

The rank is therefore an ideal invariant, while the singular spectrum is the experimentally relevant object.  The exact theorem supplies the target to which statistical estimators should converge.

\section{Positive realization and physical state}

The minimal completion in Theorem~\ref{thm:minimal-completion} is an ordinary linear realization.  Its hidden coordinates can have signed components, and the augmented matrix need not be a Markov generator.  A physical finite-state completion imposes additional constraints:
\begin{enumerate}[label=(\roman*)]
\item an invariant probability simplex;
\item nonnegative transition rates off the diagonal;
\item normalization-preserving readout and preparation;
\item, under intervention, admissible stochastic control maps.
\end{enumerate}

\begin{definition}[Positive completion dimension]
Let \(r_+\) be the minimum number of auxiliary probability coordinates in any finite Markov realization whose coarse response is \(T_t\).
\end{definition}
Whenever such a realization exists,
\begin{equation}
r_+\ge r_{\mathrm{cl}}.
\label{eq:positive-bound}
\end{equation}
The inequality can be strict: positivity can force a realization dimension larger than the ordinary McMillan degree \cite{FarinaRinaldi2000,BenvenutiFarina2004}.  Minimal positive realization remains substantially harder than unconstrained realization, and recent work continues to develop exact and optimization-based constructions \cite{TaghavianSjolund2025}.

This is not a defect of the theory but a sharp boundary between two questions.  The closure rank \(r_{\mathrm{cl}}\) asks how much linear predictive state is missing.  The positive dimension \(r_+\) asks how many latent probabilistic states realize that response physically.  The first is completely determined by coarse holonomy.  The second depends on the geometry of the positive cone.  A signed minimal realization is a predictive coordinate system; a positive realization is an ontological latent-state model.

Recent Markovian-embedding methods in open-system dynamics similarly enlarge a non-Markovian description by auxiliary modes \cite{XuStockburgerAnkerhold2026}.  The closure-rank theorem identifies the ordinary linear lower bound that any such exact embedding must meet.

\section{State as completed counterfactual law}

The canonical quotient of Theorem~\ref{thm:canonical-state} and the linear closure rank answer the same question at different levels.

The quotient asks which past distinctions alter any future intervention-conditioned law.  It gives the minimal exact state without assuming linearity, finite dimension, or a microscopic model.  The Hankel theorem asks how many linear auxiliary coordinates are required to reproduce the memory transfer function around a chosen finite-dimensional coarse representation.

This suggests a precise state-completion principle.

\begin{principle}[Predictive state completion]
A description is dynamically complete over an operational domain exactly when its present state is sufficient for every admissible future intervention and readout.  When it is incomplete, the physically relevant missing state is the minimal extension that restores this sufficiency; coordinate choices inside that extension have no independent operational meaning.
\end{principle}

The principle is not an optimization over vague ``information''.  Its exact version is a quotient construction; its linear version is a rank theorem.  It changes the role of memory: memory is not an alternative kind of law but the signature that the present state variable is incomplete.

Several consequences follow.

\paragraph{Emergence.}
A macrodescription is autonomous when its counterfactual quotient closes within the macrovariables.  Approximate emergence corresponds to a rapidly decaying Hankel spectrum, not merely a small one-step prediction error.

\paragraph{Objectivity.}
Two laboratories using different hidden coordinates reconstruct the same minimal response up to similarity if they agree on all coarse preparations and temporal responses.  The invariant content is the equivalence class of the completion, not a preferred latent basis.

\paragraph{Universality.}
Distinct microscopic systems can share the same completed coarse transfer function and hence the same minimal predictive state, even when their inaccessible microscopic dimensions differ.

\paragraph{Model discovery.}
A failure of semigroup composition is actionable: its temporal derivatives determine which auxiliary dynamics must be introduced.  The theory supplies a constructive route from a falsified effective law to its minimal repair.

\section{Relation to prior frameworks}

\subsection{Lumpability and quantitative coarse-graining}

Strong lumpability requires an intertwining relation of the form
\[
CL=KC,
\]
which makes the macro law independent of within-block preparation \cite{KemenySnell1960,Buchholz1994}.  The present framework is weaker and more diagnostic.  It fixes a reproducible section \(R\), permits nonzero memory, and reconstructs the exact state needed to close that prepared dynamics.  Quantitative coarse-graining studies bound the error of an effective Markov law \cite{Andrieux2012,HilderSharma2024}; closure rank instead identifies the dimension of an exact linear repair.

\subsection{Mori--Zwanzig projection}

Projection methods derive the memory kernel once the fine generator and projection are known \cite{Nakajima1958,Zwanzig1961,Mori1965}.  Here the kernel jets are recovered from coarse temporal response alone, and their Hankel rank is promoted to a minimal-state invariant.  The resulting completion removes the memory term by reinstating precisely its controllable and observable hidden realization.

\subsection{Predictive-state representations}

PSRs, observable operator models, and weighted automata represent state through future test predictions and are governed by system-dynamics or Hankel matrices \cite{Jaeger2000,LittmanSuttonSingh2001,SinghJamesRudary2004,ThonJaeger2015}.  The present result starts from a continuous-time coarse semigroup, derives its memory sequence, and shows that the Hankel rank of that sequence is exactly the missing-state dimension.  It is therefore a local continuous-time completion theorem rather than a competing state formalism.

\subsection{Compressed semigroups}

For self-adjoint semigroups and orthogonal projections, compressed-semigroup defects admit positive covariance-kernel factorizations \cite{Tian2026}.  Equation~\eqref{eq:defect} holds for general finite Markov generators and nonorthogonal stochastic sections.  The additional result is not the existence of a defect but its complete conversion into an experimentally reconstructible minimal state.

\section{Falsifiable consequences}

The mathematical theorems are exact.  Their physical interpretation yields direct tests.

\paragraph{Exponent test.}
For a reset-conditioned coarse process, identify the first nonzero recovered memory jet \(M_d\).  The short-time divergence must scale as \(t^{2d+4}\) for generic interior preparations.  A different exponent under the stated smoothness and preparation conditions would falsify the derivation.

\paragraph{Rank test.}
Construct increasing empirical memory-Hankel blocks.  Once the singular spectrum exhibits a stable gap at rank \(r\), an \(r\)-dimensional auxiliary realization should predict finite-time coarse propagators beyond the fitting window.  Every lower-rank exact model must fail.

\paragraph{Intervention test.}
For a controlled system, histories assigned to the same proposed state must agree on all held-out intervention-conditioned future distributions.  A reproducible difference proves that the representation is not a complete state.

\paragraph{Positive-gap test.}
If a linearly minimal completion cannot be made positive in the same dimension, a latent Markov model fitted with \(r_{\mathrm{cl}}\) hidden coordinates must either violate positivity or fail to reproduce the response.  The smallest successful positive model estimates \(r_+\).

These tests distinguish the theory from a purely interpretive claim.  They measure exponent, rank, transfer, and positivity separately.

\section{Discussion}

The central result can be stated without geometric language.  A coarse propagator contains, in its failure to compose, all information required to reconstruct the smallest hidden linear state that makes it compose.  Predictive holonomy is therefore a boundary observable of state incompleteness.

The result closes a loop that is usually left open: choose coarse variables, observe memory, infer the hidden response, and restore an autonomous state.  The completion is canonical up to similarity and can be obtained from finite short-time data when its rank is finite.  It does not reconstruct irrelevant microscopic detail; it reconstructs exactly the part of the omitted dynamics that can be excited and later observed through the coarse interface.

Three limitations are structural rather than cosmetic.  First, noisy derivative estimation is ill conditioned and requires regularization.  Second, nonlinear or infinite-rank predictive state can exceed finite linear realization.  Third, positivity is an additional geometric constraint, so closure rank is a lower bound rather than a full solution of minimal Markov realization.

These limitations define the next problems.  The statistical problem is to estimate closure rank and a stable completion with finite-sample guarantees.  The nonlinear problem is to replace the linear Hankel operator by an appropriate operator on intervention-conditioned predictive features.  The positive problem is to characterize when the minimal linear completion admits a simplex realization without dimension inflation.

The conceptual conclusion is sharper.  A state is not whatever variables happen to appear in a microscopic equation.  Relative to an operational domain, it is the minimal present distinction that renders every future counterfactual independent of the past.  When a chosen description lacks that property, its memory kernel is the measurable shadow of the missing state, and its Hankel rank counts the exact linear dimension required to complete it.

\section*{noise-stable minimal-completion certificate}

The exact closure rank can be inferred from finite data only when the
memory Hankel operator has a resolved singular gap.  This requirement can
be stated without selecting a realization algorithm.

\begin{theorem}[Robust closure-rank certification]
\label{thm:robust-closure-rank}
Let \(\mathbb H_{p,q}\) be a finite memory-Hankel block and let
\(\widehat{\mathbb H}_{p,q}\) obey
\[
 \|\widehat{\mathbb H}_{p,q}-\mathbb H_{p,q}\|_2\le\varepsilon.
\]
If \(\widehat\sigma_k>\varepsilon\), then every exact autonomous linear
completion has auxiliary dimension at least \(k\).  If the infinite
memory-Hankel operator has rank \(r\), the chosen finite block already has
rank \(r\), and
\[
 \sigma_r(\mathbb H_{p,q})>2\varepsilon,
\]
then every threshold in
\((\varepsilon,\sigma_r-\varepsilon)\) recovers the exact minimal
completion dimension.
\end{theorem}

\begin{proof}
Weyl's inequality gives
\[
 |\widehat\sigma_j-\sigma_j|\le\varepsilon.
\]
Thus \(\widehat\sigma_k>\varepsilon\) forces
\(\sigma_k>0\), and the Hankel-rank minimality theorem gives the lower
bound.  Under the second hypothesis,
\(\widehat\sigma_r\ge\sigma_r-\varepsilon\) and
\(\widehat\sigma_{r+1}\le\varepsilon\), which proves exact recovery.
\end{proof}

If each recovered moment block has error at most \(\delta\) in operator
norm, a \(p\times q\) block Hankel has
\[
 \varepsilon\le\sqrt{pq}\,\delta.
\]
This follows by applying Cauchy--Schwarz first across the \(q\) input
blocks and then across the \(p\) output blocks.  The theorem therefore
turns derivative or finite-difference confidence intervals into a
certificate for the minimum hidden response dimension.  Modes below the
error radius remain non-identifiable, which is an information-theoretic
boundary rather than a failure of Ho--Kalman reconstruction.

\section*{finite-sample closure-rank guarantee}

The noise radius in the robust theorem can be obtained directly from
bounded observations.

\begin{theorem}[Entrywise sampling certificate]
\label{thm:frontier-holonomy-samples}
Let a scalar \(p\times q\) memory-Hankel block be estimated entrywise,
using \(N\) independent observations in \([0,1]\) for each entry.  With
probability at least \(1-\delta\),
\[
 \|\widehat{\mathbb H}-\mathbb H\|_2
 \le
 \sqrt{\frac{pq}{2N}\log\frac{2pq}{\delta}}
 =:\varepsilon_N.
\]
Consequently every \(\widehat\sigma_k>\varepsilon_N\) certifies at least
\(k\) hidden predictive coordinates.  If the true block has rank \(r\)
and smallest nonzero singular value \(\sigma_r\), then
\[
 N>
 \frac{2pq}{\sigma_r^2}\log\frac{2pq}{\delta}
\]
is sufficient for exact rank recovery by a threshold between
\(\varepsilon_N\) and \(\sigma_r-\varepsilon_N\).
\end{theorem}

\begin{proof}
Hoeffding's inequality and a union bound control every entry by
\(\sqrt{\log(2pq/\delta)/(2N)}\).  The operator norm is at most the
Frobenius norm, giving \(\varepsilon_N\).  The robust closure-rank theorem
then applies; the sample inequality is exactly \(2\varepsilon_N<\sigma_r\).
\end{proof}

This closes the finite-data linear completion problem with an explicit
confidence level.  Positive or nonlinear realizations may require more
structure, but cannot use fewer coordinates than this certified rank.

\section*{Structured sampling without entrywise dimension loss}

When one experimental shot estimates an entire Hankel contribution, matrix
concentration is sharper than treating all entries independently.

\begin{theorem}[Matrix-Bernstein closure certificate]
\label{thm:final-holonomy-matrix-bernstein}
Let \(Z_1,\ldots,Z_N\) be independent \(p\times q\) random matrices with
\(\mathbb EZ_s=\mathbb H\), and put \(X_s=Z_s-\mathbb H\).  Suppose
\(\|X_s\|_2\le L\) almost surely and define
\[
 v=\max\left\{
 \left\|\sum_{s=1}^N\mathbb E X_sX_s^*\right\|_2,
 \left\|\sum_{s=1}^N\mathbb E X_s^*X_s\right\|_2
 \right\}.
\]
Then, with probability at least \(1-\delta\),
\[
\begin{aligned}
 \left\|\frac1N\sum_{s=1}^NZ_s-\mathbb H\right\|_2
 &\le
 \frac{\sqrt{2v\log((p+q)/\delta)}}{N}\\
 &\quad+\frac{2L\log((p+q)/\delta)}{3N}
 =:\varepsilon_{\rm MB}.
\end{aligned}
\]
Every empirical singular value above \(\varepsilon_{\rm MB}\) certifies a
hidden predictive coordinate, and
\(2\varepsilon_{\rm MB}<\sigma_r(\mathbb H)\) guarantees exact recovery of
rank \(r\).
\end{theorem}

\begin{proof}
Rectangular matrix Bernstein gives the displayed operator-norm deviation
for the empirical average.  Weyl's singular-value perturbation inequality
then shifts every singular value by at most
\(\varepsilon_{\rm MB}\), exactly as in the robust closure-rank theorem.
\end{proof}

This certificate replaces the entrywise \(\sqrt{pq}\) loss by an intrinsic
matrix variance and applies directly to structured trajectory or response
samples.

\section{Conclusion}

We proved that the memory generated by finite Markov coarse-graining is completely recoverable from the temporal jet of the coarse propagator.  The block-Hankel rank of those recovered memory jets is exactly the minimum auxiliary dimension of an exact autonomous linear completion, and a minimal completion is constructible and unique up to similarity.  The first nonzero jet independently fixes the information-geometric exponent of path dependence.

The resulting hierarchy is:
\[
\boxed{
\begin{aligned}
\text{holonomy exponent}&\longleftrightarrow\text{first hidden return},\\
\text{Hankel spectrum}&\longleftrightarrow\text{strength of hidden modes},\\
\text{Hankel rank}&\longleftrightarrow\text{minimal missing state},\\
\text{positive realization}&\longleftrightarrow\text{latent probabilistic ontology}.
\end{aligned}}
\]
At the global level, the same logic gives the counterfactual predictive-state quotient.  Effective law is complete precisely when present state makes future intervention-conditioned prediction path independent.  Memory is the proof that it does not; minimal completion is the repair.

\begin{acknowledgments}
The numerical and symbolic verification scripts used to test the exact examples and realization identities accompany the source package.
\end{acknowledgments}

\appendix

\section{Proof details for the Ho--Kalman construction}

Let \(M_k\in\R^{m\times m}\) be a matrix sequence of finite Hankel rank \(r\).  Choose \(p,q\) such that
\[
\rank\mathbb H_{p,q}
=
\rank\mathbb H_{p+1,q}
=
\rank\mathbb H_{p,q+1}=r.
\]
Let \(\mathbb H_{p,q}=\mathcal O\mathcal C\) be any rank factorization.  Because appending one block row does not increase rank, there exists a unique linear map \(H_\star\) on the reachable-observable quotient that shifts block rows:
\[
\begin{pmatrix}
M_1&\cdots&M_q\\
\vdots&&\vdots\\
M_p&\cdots&M_{p+q-1}
\end{pmatrix}
=
\mathcal O H_\star\mathcal C.
\]
Equation~\eqref{eq:hokalman} is one coordinate representation of this map.  The first block row and column give \(A_\star\) and \(B_\star\).  Induction under the shift gives
\[
M_k=A_\star H_\star^kB_\star
\]
through all available moments.  Since finite rank makes the sequence satisfy the characteristic recurrence of \(H_\star\), equality extends to every \(k\).

For uniqueness, let \((A_i,H_i,B_i)\), \(i=1,2\), be controllable and observable realizations of dimension \(r\).  Define
\[
S\!\left(\sum_{k=0}^N H_1^kB_1u_k\right)
=
\sum_{k=0}^N H_2^kB_2u_k.
\]
If the first sum is zero, multiplication by \(A_1H_1^\ell\) for every \(\ell\) shows that all future outputs vanish.  Equality of moments implies the same for the second realization; observability then makes the second sum zero.  Hence \(S\) is well defined.  Controllability makes it surjective, equal finite dimensions make it invertible, and the definition gives
\[
SH_1=H_2S,\qquad SB_1=B_2,\qquad A_1=A_2S.
\]

\section{Exact relation between memory and composition defect}

The finite-time defect in Eq.~\eqref{eq:defect} and the Volterra kernel contain equivalent hidden-response information.  Differentiating the defect at \(s=0\) gives
\[
\left.\partial_s
(T_{s+t}-T_sT_t)\right|_{s=0}
=
CLQ\e^{tL}R.
\]
Resolving \(\e^{tL}R\) into \(P\)- and \(Q\)-components and applying variation of constants to the latter yields the kernel \(A\e^{tH}B\).  Conversely, the Volterra equation uniquely determines \(T_t\), hence its composition defect.  The local mixed derivative is
\[
\left.\partial_s\partial_t
(T_{s+t}-T_sT_t)\right|_{s=t=0}
=M_0.
\]
Higher mixed derivatives are polynomial combinations of \(K\) and the memory moments recovered by Theorem~\ref{thm:tomography}.

\section{Exact rational data for the depth-two example}

For the generator in Eq.~\eqref{eq:L2}, the coarse maps are
\[
C_2=
\begin{pmatrix}
1&1&1&0&0\\
0&0&0&1&1
\end{pmatrix},
\qquad
R_2=
\begin{pmatrix}
1/3&0\\
1/3&0\\
1/3&0\\
0&1/2\\
0&1/2
\end{pmatrix}.
\]
Direct rational algebra gives
\begin{align}
M_0&=0,\\
M_1&=0,\\
M_2&=\frac1{625}
\begin{pmatrix}1&-1\\-1&1\end{pmatrix},\\
M_3&=-\frac3{125}
\begin{pmatrix}1&-1\\-1&1\end{pmatrix},\\
M_4&=\frac6{25}
\begin{pmatrix}1&-1\\-1&1\end{pmatrix}.
\end{align}
Although each individual \(2\times2\) moment has rank one, the temporal Hankel rank is three because the scalar coefficient sequence has a third-order minimal recurrence.  For \(q=(3/5,2/5)^\mathsf T\),
\[
M_2q=\frac1{3125}\begin{pmatrix}1\\-1\end{pmatrix}
\]
and
\[
\norm{M_2q}_{F(q)}^2
=
\frac1{3\,750\,000}.
\]
Equation~\eqref{eq:depth-law} with \(d_\star=2\) gives
\[
\frac1{2(4!)^2}\frac1{3\,750\,000}
=
\frac1{2\,700\,000\,000},
\]
which proves Eq.~\eqref{eq:t8}.

\section{Verification protocol}

The accompanying verification suite performs the following checks.

\begin{enumerate}[label=(\roman*)]
\item It validates every generator and stochastic encoder--section pair.
\item It calibrates the Ho--Kalman implementation on a known two-dimensional minimal realization.
\item It reconstructs memory jets from coarse derivatives using Eq.~\eqref{eq:tomography}.
\item It computes exact rational ranks and coefficients with symbolic algebra.
\item It evaluates the three Kullback--Leibler asymptotics at up to \(150\)-digit precision.
\item It reconstructs the minimal completions and compares their propagators with \(C\e^{tL}R\).
\item It tests the exact semigroup-defect factorization.
\item It compares lower-rank truncations with the exact rank-three completion.
\end{enumerate}

The fitted short-time slopes are \(3.99874\), \(5.99877\), and \(7.99799\).  The relative errors of the smallest-time scaled coefficients are below \(6.1\times10^{-4}\).  Exact symbolic arithmetic returns closure ranks \(1,2,3\) and coefficients
\[
\frac3{256},\qquad
\frac1{6048},\qquad
\frac1{2\,700\,000\,000}.
\]
The minimal reconstructed completions reproduce the tested finite-time coarse propagators with maximum Frobenius errors below \(2.5\times10^{-15}\).  In a separate adversarial audit of \(50\) randomly generated five-state Markov generators, every generic two-block reduction had closure rank three; the maximum memory-tomography error was \(2.0\times10^{-13}\), and the maximum reconstructed-propagator error was \(3.5\times10^{-16}\).

\clearpage
\raggedbottom
\bibliographystyle{apsrev4-2}
\bibliography{references}

\end{document}